\documentclass[conference]{IEEEtran}
\IEEEoverridecommandlockouts
\usepackage{cite}
\usepackage{amsmath,amssymb,amsfonts}
\usepackage{algorithmic}
\usepackage{graphicx}
\usepackage{textcomp}
\usepackage{xcolor}
\usepackage{placeins}
\usepackage{braket}
\usepackage{amsthm}
\newtheorem{theorem}{Theorem}
\usepackage[caption=false,font=footnotesize]{subfig}
\def\BibTeX{{\rm B\kern-.05em{\sc i\kern-.025em b}\kern-.08em
    T\kern-.1667em\lower.7ex\hbox{E}\kern-.125emX}}
\begin{document}

\title{Entanglement Rate Maximization for Dual-Connectivity Wireless Quantum Networks
\\
{\footnotesize}

}
\author{
\IEEEauthorblockN{
Kavini Thenuwara, Shiva Kazemi Taskooh, and Ekram Hossain
}
\IEEEauthorblockA{
Department of Electrical and Computer Engineering,
University of Manitoba, Canada\\
Emails: thenuwak@myumanitoba.ca, shiva.kazemitaskooh@umanitoba.ca, ekram.hossain@umanitoba.ca
}
}

\maketitle

\begin{abstract}
The development of quantum networks (QNs) relies on efficient mechanisms for distributing entanglement among multiple quantum users (QUs) under practical system constraints. This paper investigates the problem of entanglement rate maximization in a dual‑connectivity (DC) wireless quantum network comprising multiple quantum base stations (QBSs). Under the DC architecture, each QU can associate with up to two QBSs, thereby enhancing resource utilization compared to conventional single‑connectivity (SC) schemes. The joint QBS–QU association and entanglement generation rate allocation problem is formulated as a mixed‑integer nonlinear programming problem that incorporates practical constraints, including limited QBS entanglement generation capacity as well as heterogeneous minimum entanglement rate demands  and fidelity requirements for QUs. To efficiently solve this challenging problem, an alternating optimization (AO) algorithm is developed, which decomposes the original formulation into entanglement rate allocation and association subproblems. Simulation results demonstrate that the proposed DC architecture significantly outperforms SC schemes, while the AO algorithm achieves near‑optimal performance with substantially reduced computational complexity.
\end{abstract}


\section{Introduction}
Recent advances in quantum networks (QNs) have intensified the demand for communication technologies capable of interconnecting multiple quantum devices. Quantum communication (QC) enables the exchange of quantum information among such devices and constitutes a fundamental building block for future QNs and the envisioned quantum internet \cite{dutta2024quantum}. Beyond connectivity, QC provides fundamental advantages over classical communication, including enhanced security guarantees and improved information-theoretic performance, positioning QNs as a key enabler of next-generation communication systems \cite{zhou2025towards}. These networks support a wide range of applications such as quantum key distribution (QKD), distributed quantum computing, quantum sensing, and quantum teleportation. At the core of many quantum applications lies quantum entanglement, a uniquely quantum resource in which the states of two or more particles become intrinsically correlated irrespective of the distance separating them, enabling non-classical information sharing between remote nodes \cite{einstein1935can}. Consequently, the ability to efficiently distribute entanglement across distant nodes is a central challenge in the design and operation of large-scale QNs.

Entanglement distribution in QNs can be realized using either optical fiber links or free-space optical (FSO) communication channels \cite{djordjevic2020global}. Most existing quantum network architectures predominantly rely on optical fiber to transmit and share quantum information between network nodes \cite{chehimi2025entanglement,chehimi2021entanglement,panahi2025energy}. However, the deployment of fiber infrastructure can be challenging or impractical in regions characterized by harsh terrain, physical obstacles, or limited communication infrastructure. In such scenarios, FSO quantum channels offer a flexible and cost-efficient alternative for enabling quantum communication applications \cite{pirandola2021limits}. In FSO-based systems, quantum optical signals are transmitted wirelessly through the atmosphere between terrestrial stations or via satellite links in space. As a result, FSO-based quantum networks provide greater deployment flexibility than their fiber-based counterparts, while also enabling high-capacity optical transmission by operating over a broad range of wavelengths and frequencies \cite{chehimi2025reconfigurable}.

Several studies have investigated entanglement distribution in QNs, mainly in satellite-based scenarios, where long-distance FSO quantum communication is enabled through the generation of entangled photon pairs. For instance, the authors in \cite{panigrahy2022optimal} have focused on optimizing satellite-to-ground transmissions to maximize entanglement distribution rates, while \cite{wei2024optimal} has improved entanglement distribution strategies for efficient quantum communication. In addition, the authors of \cite{wei2024entanglement} have proposed a hybrid quantum–classical algorithm to optimize entanglement rate, satellite assignment, and routing in large-scale QNs while satisfying minimum fidelity requirements of users.
By contrast, only a limited number of studies have addressed terrestrial FSO-based QNs \cite{10313846,ntanos2021availability,9573460}, with most of them concentrating on QKD over FSO channels and analyzing the impact of atmospheric turbulence, channel fading, and pointing errors on communication performance. In particular, the authors of \cite{chehimi2025reconfigurable} have investigated reconfigurable intelligent surface (RIS)-assisted FSO-based QNs to overcome line-of-sight blockages and optimize entanglement distribution under realistic channel impairments, while jointly optimizing RIS placement and entanglement generation rate allocation to satisfy users’ entanglement rate and fidelity requirements.

%Despite the advantages of FSO communication, the realization of FSO-based QNs often relies on centralized network architectures, which introduce several significant challenges. 
Most existing works on FSO-based QNs adopt a star-shaped topology, in which one or multiple central QBSs (either terrestrial QBSs or satellites) distribute entangled qubits to multiple connected QUs, while each QU is restricted to associating with only one QBS through single-connectivity (SC). Due to the inherently limited entanglement generation capacity of a QBS, this SC approach can constrain the network’s ability to satisfy the aggregated entanglement demands of all QUs. Furthermore, entangled quantum states typically exhibit short coherence times, which are affected by decoherence in quantum memory units at the QBS as well as losses and impairments during transmission over FSO channels, ultimately leading to reduced entanglement fidelity. Hence, assigning appropriate QBSs to each QU in a manner that preserves minimum fidelity requirements of QUs becomes a critical design challenge in FSO-based QNs.

To address these limitations, in this paper we develop a novel FSO-based dual-connectivity (DC) enabled QN that optimizes entanglement generation rates for heterogeneous QUs while jointly determining their association with multiple QBSs. Specifically, each QU is intelligently assigned to up to two QBSs based on its minimum fidelity and entanglement rate requirements, while accounting for entanglement generation capacity constraints of QBSs. This DC architecture alleviates the bottlenecks of SC systems, which often fail to satisfy the stringent rate requirements of quantum applications such as superdense coding and distributed quantum computing. Moreover, the use of multiple QBSs enhances network reliability and robustness, enabling more consistent entanglement distribution in the presence of decoherence and environmental impairments. 
The major contributions of this work are summarized as follows:
\begin{itemize}
    \item We propose an FSO-based QN in which QUs have heterogeneous fidelity and entanglement rate requirements, reflecting the demands of different quantum applications. Also, in the proposed DC architecture, each QU can simultaneously receive entangled qubits from up to two QBSs. In addition, we develop a comprehensive system model that captures the effects of FSO channel impairments, including atmospheric loss, turbulence, and pointing errors, as well as fidelity degradation caused by decoherence in quantum memory units.

    \item The entanglement rate maximization problem is formulated as a mixed-integer nonlinear programming (MINLP) problem, subject to per-QU minimum entanglement rate and fidelity constraints, QBS entanglement generation capacity limits, and QU association constraints. To address this problem, we adopt an alternating optimization (AO) framework that decomposes the original MINLP into two coupled subproblems corresponding to user association and entanglement generation rate allocation.

    \item Simulation results demonstrate the effectiveness of the proposed approach, showing that the AO method achieves near-optimal performance with optimality gaps of $5\%$--$19\%$, while the DC architecture improves the entanglement rate by $19.5\%$--$37\%$ compared with the SC scheme.
\end{itemize}

The rest of this paper is organized as follows. Section~\ref{sec:2} introduces the system model and assumptions. Section~\ref{sec:3} formulates the entanglement rate maximization problem and presents the proposed solution approach. Simulation results are discussed in Section~\ref{sec:4}, followed by conclusion in Section~\ref{Conclusion}.

\section{System Model and Assumptions}
\label{sec:2}
\subsection{Network Model}
\begin{figure}[t]
\centering
\includegraphics[width=0.67\linewidth]{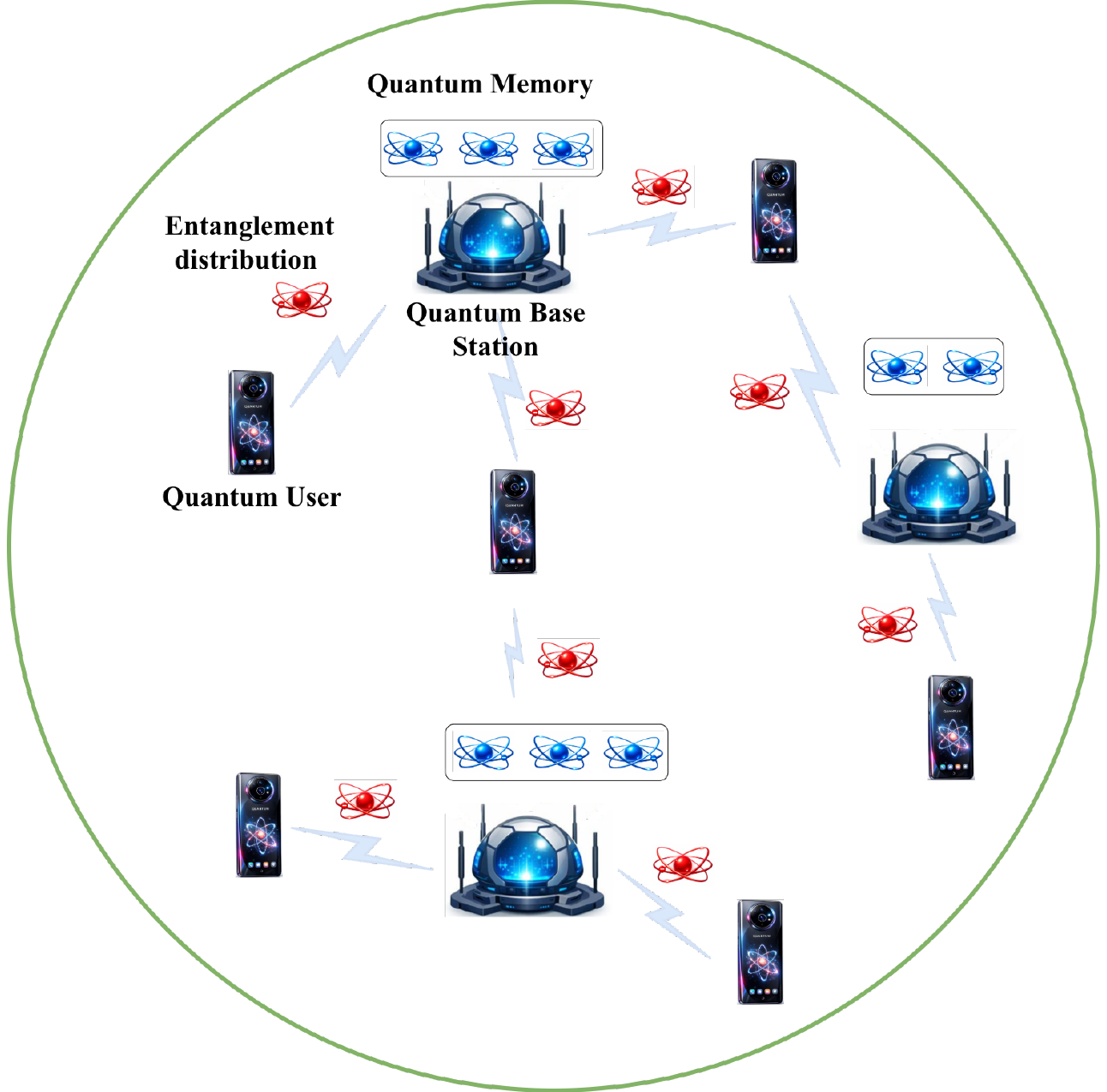}
\caption{A schematic view of the considered quantum network.}
\label{fig:network_model}
\end{figure}
\label{sec:fso_channel_model}
%We consider a dual-connectivity-enabled wireless quantum network in which two quantum base stations (QBSs) simultaneously establish entanglement with a set of $\mathcal{U}=\{1,2,\cdots, U\}$ quantum users (QUs). Each QBS generates entangled pairs, stores one half in its quantum memory, and transmits the other half to the QU via a free-space optical (FSO) channel. Both QBS consist of a quantum memory in each, with a maximum qubit capacity $R_n^{\mathrm{max}}$ where $n$ is the number of QBS.

We consider a DC-enabled wireless quantum network, as illustrated in Fig.~\ref{fig:network_model}, consisting of a set of QBSs $\mathcal{N}=\{1,2,\ldots,N\}$ and a set of heterogeneous QUs $\mathcal{U}=\{1,2,\ldots,U\}$. Each QBS, with a maximum entanglement generation capacity denoted by $R_n^{\mathrm{max}}$, generates entangled qubit pairs, stores one qubit of each pair in its local quantum memory, and transmits the other qubit to the associated QU over a FSO channel.  

The FSO channel gain from QBS~$n$ to QU~$j$ is affected by both large-scale and small-scale propagation effects of the FSO link. Specifically, the overall channel gain consists of an average gain component, which captures the effects of atmospheric loss and pointing error, and a fading gain component caused by atmospheric turbulence. In addition, the received signal strength is scaled by the responsivity of the receiver.

Atmospheric loss represents the deterministic attenuation of optical signal power caused by absorption and scattering due to atmospheric gases and weather conditions. It is modeled as follows: $g_{n,j}^l = 10^{-\kappa d_{n,j} / 10}$, 
where $\kappa$ is the weather-dependent attenuation coefficient and $d_{n,j}$ denotes the propagation distance of the FSO link between QBS~$n$ and QU~$j$ \cite{chehimi2025reconfigurable}.

Pointing error refers to the misalignment between the transmitted optical beam and the receiver aperture, which reduces the collected optical power.  Assuming a circular detection aperture, a Gaussian beam, and a zero-mean single-sided Gaussian pointing error displacement, the probability density function (PDF) of the pointing error gain $g_{n,j}^p$ is given by \cite{farid2007outage}:
\begin{equation}
\label{eq:7}
f_{g_{n,j}^{p}}\!\left(g_{n,j}^{p}\right)
= \frac{(\gamma_{n,j})^2}{A_{n,j}^{(\gamma_{n,j})^2}}
\left(g_{n,j}^{p}\right)^{(\gamma_{n,j})^2-1},
\qquad 0 \le g_{n,j}^{p} \le A_{n,j}
\end{equation}
where $A_{n,j} = [\mathrm{erf}(v_{n,j})]^2$ and
$v_{n,j} = \frac{\sqrt{\pi} a}{\sqrt{2} w_{n,j}}$ denotes the ratio between the receiver aperture radius $a$ and the beam width $w_{n,j}$. All receivers are assumed to have the same aperture radius $a$, and the beam width is given by $w_{n,j} = \theta_d d_{n,j}$. Moreover, $\gamma_{n,j} = \frac{(w_{eq}){n,j}}{2\sigma_s}$, where $\sigma_s$ denotes the standard deviation of the pointing-error displacement at the receiver and $(w_{eq})_{n,j}$ represents the equivalent beam width at the receiver.

Atmospheric turbulence in FSO channels arises from random refractive-index fluctuations in the atmosphere and is modeled using the Gamma–Gamma distribution. The PDF of the turbulence-induced fading gain $g_{n,j}^f$ is given by \cite{8601267,6786381}:
\begin{equation}
\label{eq:9}
\begin{aligned}
f_{g^{f}_{n,j}}\!\left(g^{f}_{n,j}\right)
&= \frac{2\left(\alpha_{n,j}\beta_{n,j}\right)^{\frac{\alpha_{n,j}+\beta_{n,j}}{2}}}
{\Gamma\!\left(\alpha_{n,j}\right)\Gamma\!\left(\beta_{n,j}\right)}
\left(g^{f}_{n,j}\right)^{\frac{\alpha_{n,j}+\beta_{n,j}}{2}-1}
\\
&\quad \times
K_{\alpha_{n,j}-\beta_{n,j}}
\!\left(2\sqrt{\alpha_{n,j}\beta_{n,j}g^{f}_{n,j}}\right),
\end{aligned}
\end{equation}
where $\Gamma(\cdot)$ denotes the Gamma function and $K_{\nu}(\cdot)$ denotes the modified Bessel function of the second kind. The parameters $\alpha_{n,j}$ and $\beta_{n,j}$ represent the effective numbers of small-scale and large-scale turbulence cells, respectively, and are expressed for far-field conditions as\cite{8601267}:
\begin{align}
\label{eq:10}
\alpha_{n,j} &= \left[
\exp\!\left(
\frac{0.49 (\sigma_R^2)_{n,j}}
{\left(1 + 1.11 (\sigma_R^2)_{n,j}^{12/5}\right)^{7/6}}
\right) - 1
\right]^{-1}, \\
\beta_{n,j} &= \left[
\exp\!\left(
\frac{0.51 (\sigma_R^2)_{n,j}}
{\left(1 + 0.69 (\sigma_R^2)_{n,j}^{12/5}\right)^{5/6}}
\right) - 1
\right]^{-1},
\end{align}
where $(\sigma_R^2)_{n,j} = 1.23 , C_n^2 , k^{7/6} , d_{n,j}^{11/6}$ is the Rytov variance \cite{al2001mathematical}, with $C_n^2$ denoting the refractive index structure parameter, $k = 2\pi/\lambda_{FSO}$ the optical wave number, and $\lambda_{FSO}$ the optical wavelength.
Accordingly, the overall FSO channel gain between QBS~$n$ and QU~$j$ is obtained from $g_{n,j} = \mathcal{R}   g_{n,j}^l   g_{n,j}^p   g_{n,j}^f$, where $\mathcal{R}$ denotes the receiver responsivity, assumed to be identical for all QUs.

\subsection{Entanglement Rate}
\label{sec:Fiedelity}
Let $r_{n,j}$ denote the entanglement generation rate (in pairs per second) of QBS~$n$ alloacted to QU~$j$, and let $s_{n,j}$ denote the success probability of transmitting an
entangled pair to QU~$j$. The entanglement rate at QU~$j$ from QBS~$n$ is given by $R_{n,j} = r_{n,j}\, s_{n,j}$. The probability of successful transmission depends on the FSO channel gain and can be defined as $s_{n,j} = \Pr \left( g_{n,j}\ge \eta_{\text{th}} \right)$, where $\eta_{\text{th}}$ is a predefined channel gain threshold required for
successful photon detection \cite{chehimi2025reconfigurable}. Using the Meijer G-function ($G^{x,y}_{z,m}$), a closed-form expression for this probability is obtained as follows \cite{chehimi2025reconfigurable},\cite{4432329}:
\begin{equation}
\begin{aligned}
\label{eq:success_prob}
s_{n,j}
&= 1 -
\frac{(\gamma_{n,j})^2}{\Gamma(\alpha_{n,j})\Gamma(\beta_{n,j})} \\
&\quad \times
G^{3,1}_{2,4}\!\left(
\frac{\alpha_{n,j}\beta_{n,j}\eta_{\text{th}}}
{A_{n,j} g^{l}_{n,j}\mathcal{R} }
\;\middle|\;
\begin{array}{c}
1,\, (\gamma_{n,j})^2 + 1 \\
(\gamma_{n,j})^2,\, \alpha_{n,j},\, \beta_{n,j},\, 0
\end{array}
\right).
\end{aligned}
\end{equation}

Furthermore, in the considered DC architecture, each QU can be associated with at most two QBSs, depending on its minimum entanglement rate and fidelity requirements, as well as the maximum capacity constraints of the QBSs. The QBS–QU association is represented by the binary variable $x_{n,j}$, where $x_{n,j}=1$ if QU $j$ is associated with QBS $n$, and $x_{n,j}=0$ otherwise.

\subsection{Fidelity}
In practical FSO systems, the entangled pairs shared between the two end nodes are not perfect Bell states, due to noise introduced by the channel and operations. In such situations, the quantum states of entangled pairs can be characterized as Werner states. The density matrix of the Werner state, denoted by $\rho_w$, is given by $\rho_w = W \ket{\psi^-}\bra{\psi^-} + \frac{1 - W}{4} I_4$,  
where $\ket{\psi^-} = \frac{1}{\sqrt{2}}(\ket{01} - \ket{10})$ and $I_4$ denotes the $4 \times 4$ identity matrix. Here, $W \in [0,1]$ is the Werner parameter that determines the degree of entanglement \cite{li2024narrowgap}.
The fidelity of the Werner state $\rho_w$ can be written as $F \triangleq \frac{3W + 1}{4}$. 

Let $F_{n,j}^0$ denote the fidelity of the quantum state at the channel output. In this paper, we adopt a deterministic distance-dependent model in which the channel fidelity decreases exponentially with the transmission distance, expressed as $F_{n,j}^0 = \exp(-\xi d_{n,j})$, where $\xi$ is an effective degradation coefficient that accounts for the combined effects of atmospheric attenuation, atmospheric turbulence, and pointing errors. As we consider the Werner-state representation and based on $F \triangleq \frac{3W + 1}{4}$, the channel output fidelity $F_{n,j}^0$ is converted to the corresponding Werner parameter, expressed as $W_{n,j}^0 = \frac{4F_{n,j}^{0} - 1}{3}$\cite{li2024narrowgap}.
In addition to channel impairments, quantum memory decoherence also affects the final fidelity. Memory decoherence for each user over a storage duration $\tau_{n,j}$, characterized by a coherence time $T_c$, is modeled as an exponential decay of the Werner parameter. Accordingly, the Werner parameter becomes $W_{n,j} = W_{n,j}^{0} 
\exp\!\left(-\frac{\tau_{n,j}}{T_c}\right)$, where $\tau_{n,j} = \frac{d_{n,j}}{c} + T_{\text{p}}$, in which $c$ is the speed of light and $T_{\text{p}}$ represents the processing delay~\cite{chehimi2025reconfigurable}. Consequently, the final fidelity of the entangled state delivered from QBS $n$ to QU $j$ is given by $F_{n,j} = \frac{3W_{n,j} + 1}{4}$.

\section{Problem Formulation and solution method}
\label{sec:3}
Our objective is to maximize the total entanglement rate of the network through the joint optimization of the QBS–QU association variables $x_{n,j}$ and the corresponding entanglement generation rates $r_{n,j}$. The optimization problem is formally stated as follows:
\begin{equation}\label{main-problem}
        \begin{aligned}
            &\displaystyle\max_{\substack{\{x_{n,j},r_{n,j}  \}	}}
            &&\sum\limits_{j\in\mathcal{U}} \sum\limits_{n\in\mathcal{N}}  x_{n,j} R_{n,j}  \\
            &\text{s.t.}
            &&\mathrm{C1}: \sum\limits_{n\in\mathcal{N}} x_{n,j} R_{n,j}  \geq R_j^{\mathrm{min}}, \quad \forall j\in\mathcal{U}, \\
            &
            &&\mathrm{C2}: F_{n,j}\color{black}  \geq  x_{n,j} F_j^{\mathrm{min}}, \quad \forall j\in\mathcal{U}, \quad \forall n\in\mathcal{N}, \\
            &
            &&\mathrm{C3}: \sum\limits_{j\in\mathcal{U}} x_{n,j} r_{n,j} \leq R_n^{\mathrm{max}}, \quad \forall n\in\mathcal{N}, \\
            &
            &&\mathrm{C4}: \sum\limits_{n\in\mathcal{N}}x_{n,j} \le 2, \quad \forall j\in\mathcal{U}, \\
            &
            &&\mathrm{C5}:x_{n,j} \in \{0,1\}, \quad \forall j \in \mathcal{U}, \forall n \in \mathcal{N}\\
            &
            &&\mathrm{C6}: r_{n,j} s_{n,j} \geq 0, \quad \forall j\in\mathcal{U}, \ \forall n\in\mathcal{N}.
        \end{aligned}
    \end{equation}    
Here, constraint~$\mathrm{C1}$ ensures that each user~$j$ is allocated an entanglement rate no less than its minimum required rate $R_j^{\mathrm{min}}$. Constraint~$\mathrm{C2}$ enforces the fidelity requirement by guaranteeing that the achieved entanglement fidelity $F_{n,j}$ for each user is not lower than the minimum acceptable fidelity $F_j^{\mathrm{min}}$. Constraint~$\mathrm{C3}$ limits the total entanglement generation rate allocated by QBS~$n$ to all its associated QUs to its maximum generation capacity $R_n^{\mathrm{max}}$. Moreover, constraint~$\mathrm{C4}$ guarantees that each QU~$j$ can be served by at most two QBSs. Finally, constraint~$\mathrm{C5}$ specifies the binary nature of the QBS--QU association variables, while constraint~$\mathrm{C6}$ enforces the non-negativity of the allocated entanglement generation rates.
    
The MINLP problem in~\eqref{main-problem} involves two coupled decision variables: the continuous variable $r_{n,j}$ and the binary variable $x_{n,j}$. To address this challenge, we adopt an AO approach~\cite{yi2014alternating}, in which problem~\eqref{main-problem} is decomposed into two subproblems, namely the entanglement generation rate optimization subproblem and the QBS--QU association subproblem.

\textbf{Entanglement generation rate optimization subproblem:}  Given a fixed QBS--QU association $x_{n,j}$, the entanglement generation rate optimization subproblem is formulated as follows \footnote{Note that $R_{n,j}$ is defined as $R_{n,j} = r_{n,j} s_{n,j}$; therefore, the expressions $R_{n,j}$ and $r_{n,j} s_{n,j}$ are used interchangeably throughout this paper.}:
\begin{equation}\label{subproblem-continuous}
\begin{aligned}
\max_{\{r_{n,j}\}} \quad 
& \sum_{j\in\mathcal{U}} \sum_{n\in\mathcal{N}} x_{n,j} r_{n,j} s_{n,j} \\[2pt]
\text{s.t.}\quad
& \text{C1, C3, C6}.
\end{aligned}
\end{equation}
Problem~\eqref{subproblem-continuous} is a linear optimization problem and can be solved optimally using standard optimization solvers such as the Gurobi optimizer~\cite{prajapati2025linear}.

\textbf{QBS--QU association subproblem:}
Given a fixed entanglement generation rate solution, the QBS--QU association subproblem is formulated as follows:
\begin{equation}\label{eq19}
\begin{aligned}
\max_{\{x_{n,j}\}} \quad
& \sum_{j\in\mathcal{U}} \sum_{n\in\mathcal{N}} x_{n,j} r_{n,j} s_{n,j} \\
\text{s.t.} \quad
& \mathrm{C1, C2, C3, C4, C5}.
\end{aligned}
\end{equation}
Problem~\eqref{eq19} is an integer linear programming problem that is, in general, not solvable in polynomial time. To address this challenge, following a similar approach to~\cite{proof}, we replace the binary constraint~$\mathrm{C5}$ in~\eqref{eq19} with the following equivalent set of constraints:
\begin{equation}\label{eq20}
\begin{aligned}
\mathrm{C5.1:}\quad 
& \sum_{j \in \mathcal{U}} \sum_{n \in \mathcal{N}} \left( x_{n,j} - \left(x_{n,j}\right)^2 \right) \leq 0, \\
\mathrm{C5.2:}\quad 
& 0 \leq x_{n,j} \leq 1, \quad \forall j \in \mathcal{U},\ \forall n \in \mathcal{N}.
\end{aligned}
\end{equation}

By substituting the binary constraint~$\mathrm{C5}$ in~\eqref{eq19} with constraints~$\mathrm{C5.1}$ and~$\mathrm{C5.2}$ in~\eqref{eq20}, problem~\eqref{eq19} is transformed into a non-convex optimization problem due to constraint~$\mathrm{C5.1}$. The following theorem is introduced to handle constraint~$\mathrm{C5.1}$.
\begin{theorem}
For a sufficiently large penalty factor $\lambda \gg 1$, problem~\eqref{eq19} is equivalent to the following penalized optimization problem:
\begin{equation}\label{eq21}
\begin{aligned}
\max_{\{x_{n,j}\}} \quad
& \sum_{j\in\mathcal{U}} \sum_{n\in\mathcal{N}} x_{n,j} r_{n,j} s_{n,j}
- \lambda \sum_{j \in \mathcal{U}} \sum_{n \in \mathcal{N}}
\left( x_{n,j} - \left(x_{n,j}\right)^2 \right) \\
\text{s.t.} \quad
& \mathrm{C1}, \mathrm{C2}, \mathrm{C3}, \mathrm{C4}, \mathrm{C5.2}.
\end{aligned}
\end{equation}
Here, $\lambda$ acts as a penalty factor that discourages fractional values of $x_{n,j}$, thereby enforcing binary solutions at optimality.
\end{theorem}

\begin{proof}
The proof is provided in~\cite{proof}.
\end{proof}

Let $ f(\boldsymbol{x})=\sum_{j\in\mathcal{U}} \sum_{n\in\mathcal{N}}  x_{n,j} r_{n,j} s_{n,j}- \lambda \sum\limits_{j \in \mathcal{U}}\sum\limits_{n\in\mathcal{N}} x_{n,j}$ and  $ g(\boldsymbol{x})=-\lambda\sum_{j\in\mathcal{U}}\sum_{j\in\mathcal{N}}\left({x_{n,j}}\right)^2 $. The objective function of problem~\eqref{eq21} can then be expressed as  \( f(\boldsymbol{x}) - g(\boldsymbol{x}) \), where both \( f(\boldsymbol{x}) \) and \( g(\boldsymbol{x}) \) are concave functions. Consequently, problem~\eqref{eq21} belongs to the class of difference‑of-convex programming problems~\cite{sun2016majorization}.

The majorization–minimization (MM) method is a well‑established technique to solve difference‑of‑convex problems by iteratively approximating them with a sequence of convex problems. In particular, a common MM approach employs a first‑order Taylor approximation to linearize the non‑convex component. Specifically, for a differentiable function \(h(\boldsymbol{y})\), its first‑order Taylor expansion around a feasible point \(\overline{\boldsymbol{y}}\) is given by:
$
h(\boldsymbol{y}) \approx h(\overline{\boldsymbol{y}}) + \nabla h(\overline{\boldsymbol{y}})^\top
\left(\boldsymbol{y} - \overline{\boldsymbol{y}}\right)
$, 
which provides a global affine approximation of \(h(\boldsymbol{y})\).

Accordingly, to convexify the objective function of problem~\eqref{eq21}, the non‑convex term
\( g(\boldsymbol{x}) = -\lambda \sum_{j \in \mathcal{U}} \sum_{n \in \mathcal{N}} (x_{n,j})^2 \)
is approximated using its first‑order Taylor expansion around the association vector obtained in the previous iteration, denoted by \(\boldsymbol{x}(t-1)\), as follows $g(\boldsymbol{x}) \approx g\!\left(\boldsymbol{x}(t-1)\right)
+ \nabla g\!\left(\boldsymbol{x}(t-1)\right)^\top
\left(\boldsymbol{x} - \boldsymbol{x}(t-1)\right)$. 
As a result, problem~\eqref{eq21} can be reformulated as:
\begin{equation}\label{subproblem-binary-re-taylor}
\begin{aligned}
\max_{\{x_{n,j}\}}\quad
& \sum_{j\in\mathcal{U}} \sum_{n\in\mathcal{N}} x_{n,j} r_{n,j} s_{n,j}
- \lambda \sum_{j \in \mathcal{U}} \sum_{n \in \mathcal{N}} x_{n,j} \\
& \quad + \lambda \sum_{j \in \mathcal{U}} \sum_{n \in \mathcal{N}}
\left[ 2 x_{n,j} x_{n,j}(t-1)
- \left(x_{n,j}(t-1)\right)^2 \right] \\
\text{s.t.}\quad
& \mathrm{C1, C2, C3, C4, C5.2}.
\end{aligned}
\end{equation}
Problem~\eqref{subproblem-binary-re-taylor} is a linear optimization problem, and thus can be solved efficiently in polynomial time using standard solvers such as the Gurobi optimizer \cite{prajapati2025linear}.

\section{Simulation Results}
\label{sec:4}
We evaluate the performance of the proposed DC architecture and AO method by specifying the ranges of the optimization parameters. We consider a square area of $600~\mathrm{m} \times 600~\mathrm{m}$ in which QBSs and QUs are randomly deployed. The distance of each QBS--QU link, denoted by $d_{n,j}$, is assumed to lie within the range $[150,550]~\mathrm{m}$. Unless stated otherwise, the simulation parameters are defined as follows.  
The maximum entanglement generation capacity of each QBS, $R_n^{\mathrm{max}}$, is randomly drawn from the interval $[5\times10^{6},\,1\times10^{7}]$ pairs/s. To account for heterogeneous quantum application requirements, the minimum entanglement rate demand of each QU, $R_j^{\mathrm{min}}$, varies across QUs and typically lies in the range $[2\times10^{3},\,4\times10^{3}]$ pairs/s. Moreover, the minimum required entanglement fidelity for each QU, denoted by $F_j^{\mathrm{min}}$, is independently drawn from a uniform distribution over the interval $[0.8,0.95]$. The remaining simulation parameters are summarized in Table~\ref{Simulation}. %and are adopted from~\cite{chehimi2025reconfigurable,9573460,farid2007outage,4432329}.

Fig.~\ref{fig:ao_conv} depicts the convergence behavior of the proposed AO algorithm for a single network snapshot. As can be seen, the AO method converges within $8$ iterations, demonstrating fast convergence. 

In the following, we compare the total entanglement rate achieved by the proposed DC‑based architecture with that of a conventional SC approach, as well as the performance of the proposed AO method with the optimal solution obtained using the Gurobi solver. The comparison is conducted under varying numbers of QUs, different numbers of QBSs, and different ranges of minimum entanglement rate requirements. All simulation results are obtained by averaging over $1000$ independent network snapshots, in which user locations are randomly regenerated across snapshots.

\begin{table}[t]
\caption{Summary of parameters used in simulations \cite{chehimi2025reconfigurable,9573460,farid2007outage,4432329}}
\label{Simulation}
\centering
\footnotesize
\renewcommand{\arraystretch}{1.05}
\setlength{\tabcolsep}{3pt}
\begin{tabular}{|c|p{3.7cm}|c|}
\hline
\textbf{Parameter} & \textbf{Description} & \textbf{Value} \\
\hline
$\lambda_{\mathrm{FSO}}$ & Wavelength & 1550 nm \\
\hline
$\kappa$ & Weather-dependent attenuation coefficient of the FSO link & 0.43 dB/km \\
\hline
$a$ & Aperture radius & 0.25 m \\
\hline
$\sigma_s$ & Std. dev. of the pointing-error displacement & 1 mrad \\
\hline
$\theta_d$ & Beam divergence angle & 8 mrad \\
\hline
$R$ & Receiver responsivity & 0.95 \\
\hline
$C_n^2$ & Refractive index structure constant (turbulence strength) & $5 \times 10^{-14}\ \mathrm{m}^{-2/3}$ \\
\hline
$c$ & Speed of light & $3 \times 10^8\ \mathrm{m}\mathrm{s}^{-1}$ \\
\hline
$T_c$ & Quantum memory coherence time & 2.43 ms \\
\hline
$T_p$ & Average qubit processing time & 4 $\mu\text{s}$ \\
\hline
$\eta_{\mathrm{th}}$ & FSO channel threshold & 0.05 \\
\hline
$\xi$ & Channel degradation coefficient & 0.2 dB/km \\
\hline
\end{tabular}
\end{table}

\begin{figure}[t!]
    \centering
    \includegraphics[width=0.9\linewidth]{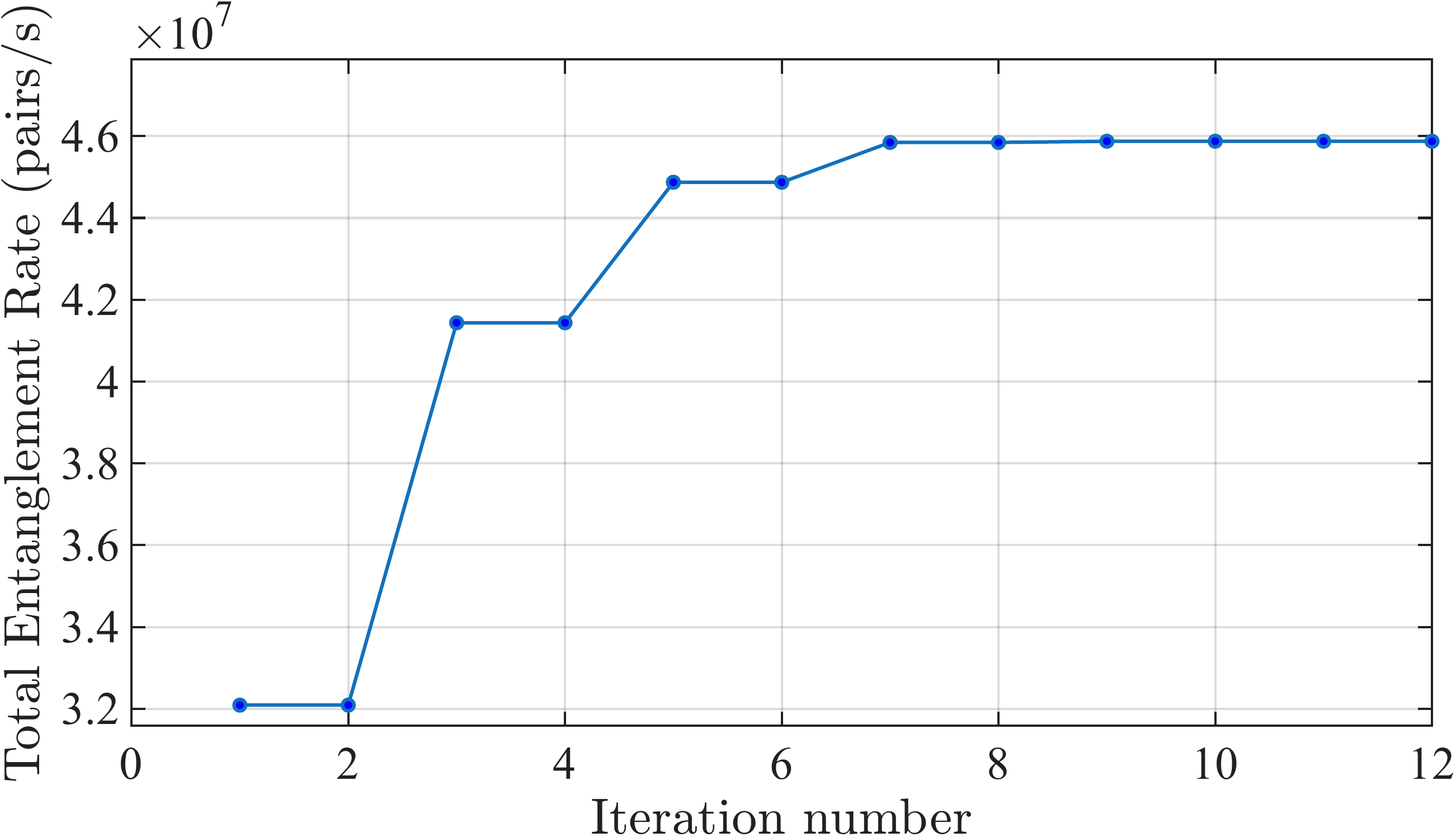}
    \caption{AO convergence for $N=10$ and $U=20$.}
    \label{fig:ao_conv}
\end{figure}

\begin{figure}
\centering
    \includegraphics[width=0.9\linewidth]{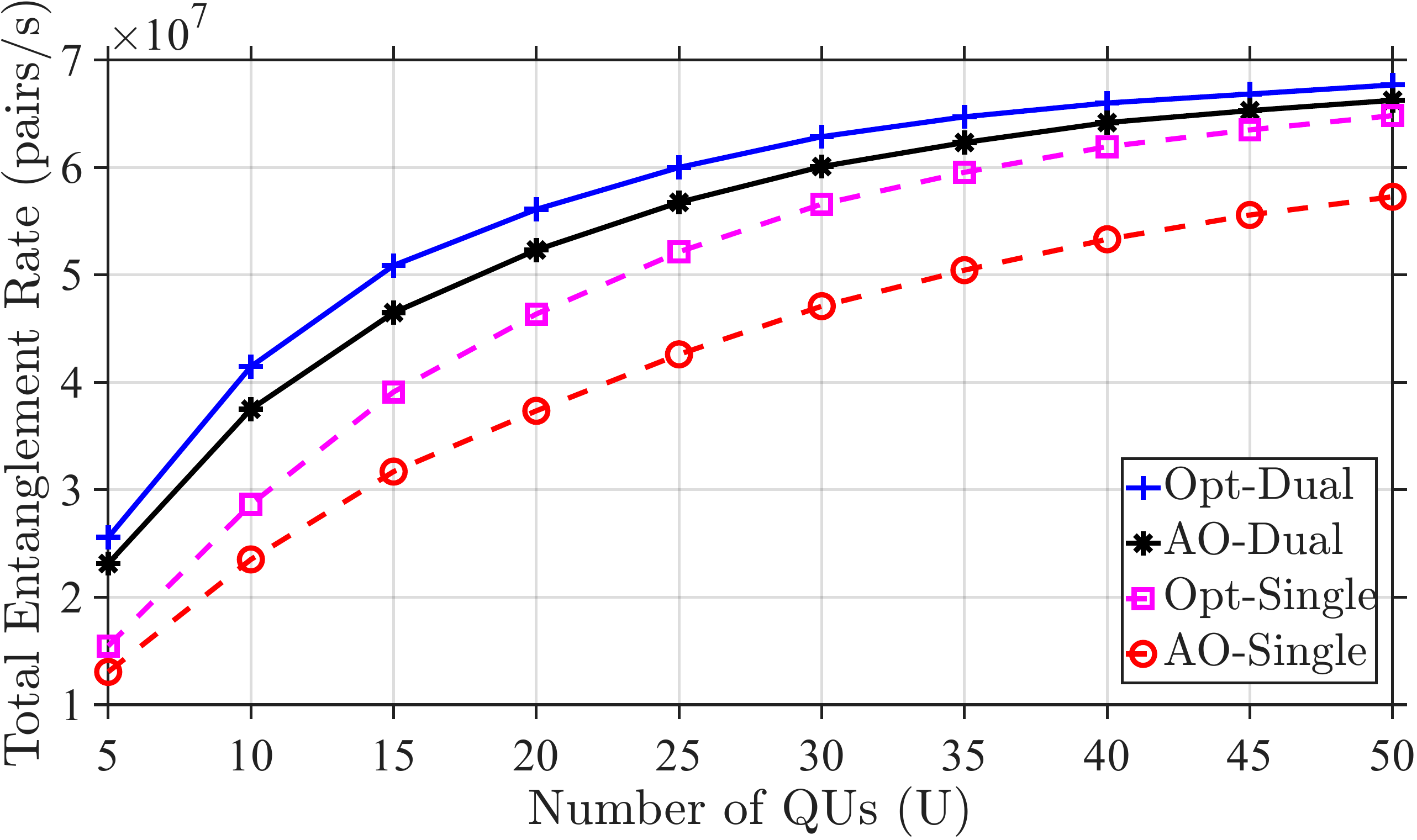}
\caption{Total entanglement rate vs. number of QUs when $N=10$.}
    \label{users}
\end{figure}
\begin{figure}
\centering
    \includegraphics[width=0.9\linewidth]{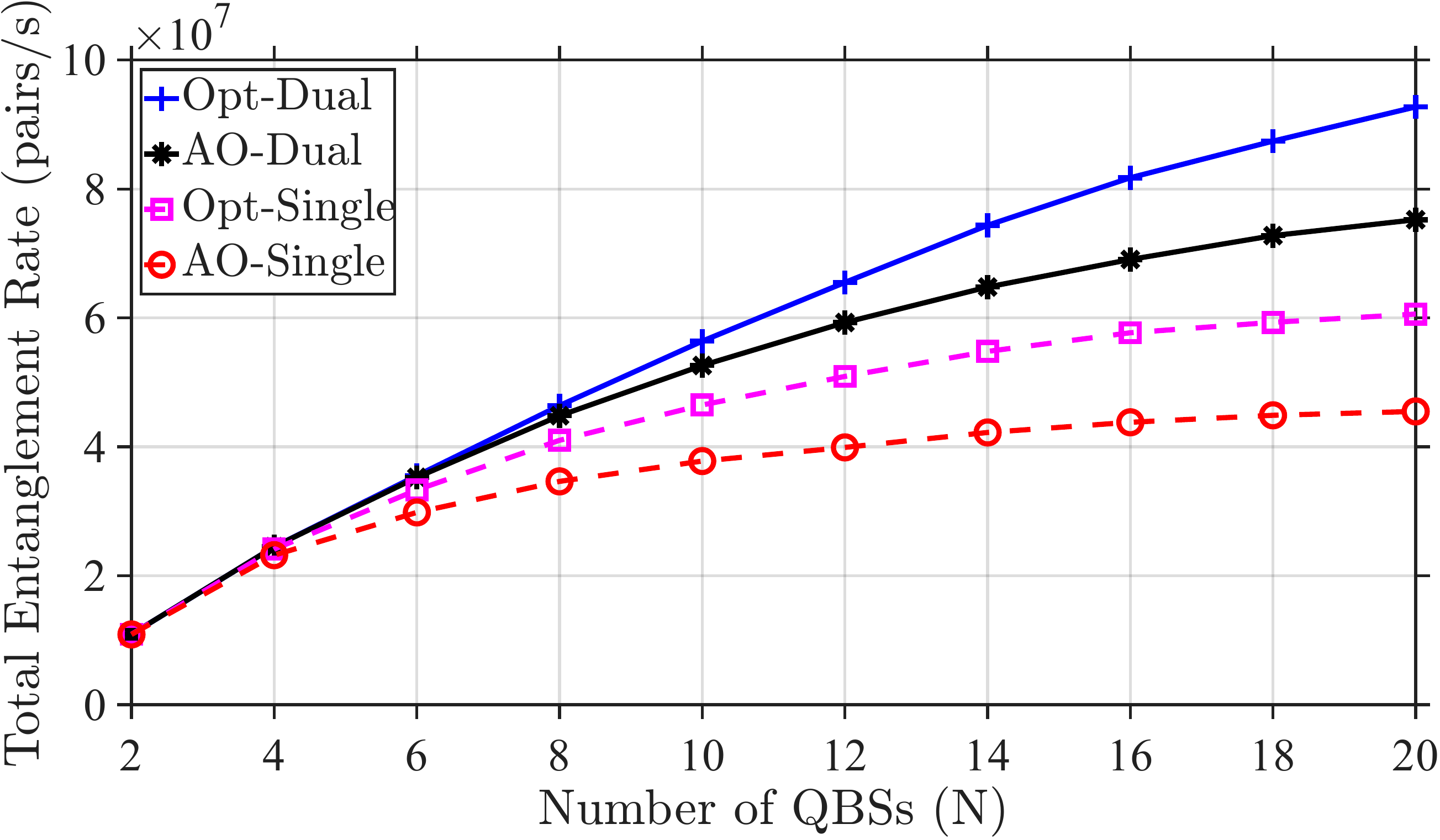}
\caption{Total entanglement rate vs.  number of QBSs  when $U=20$.}
\label{bs}
\end{figure}

\begin{figure}[t!]
\centering
    \includegraphics[width=0.9\linewidth]{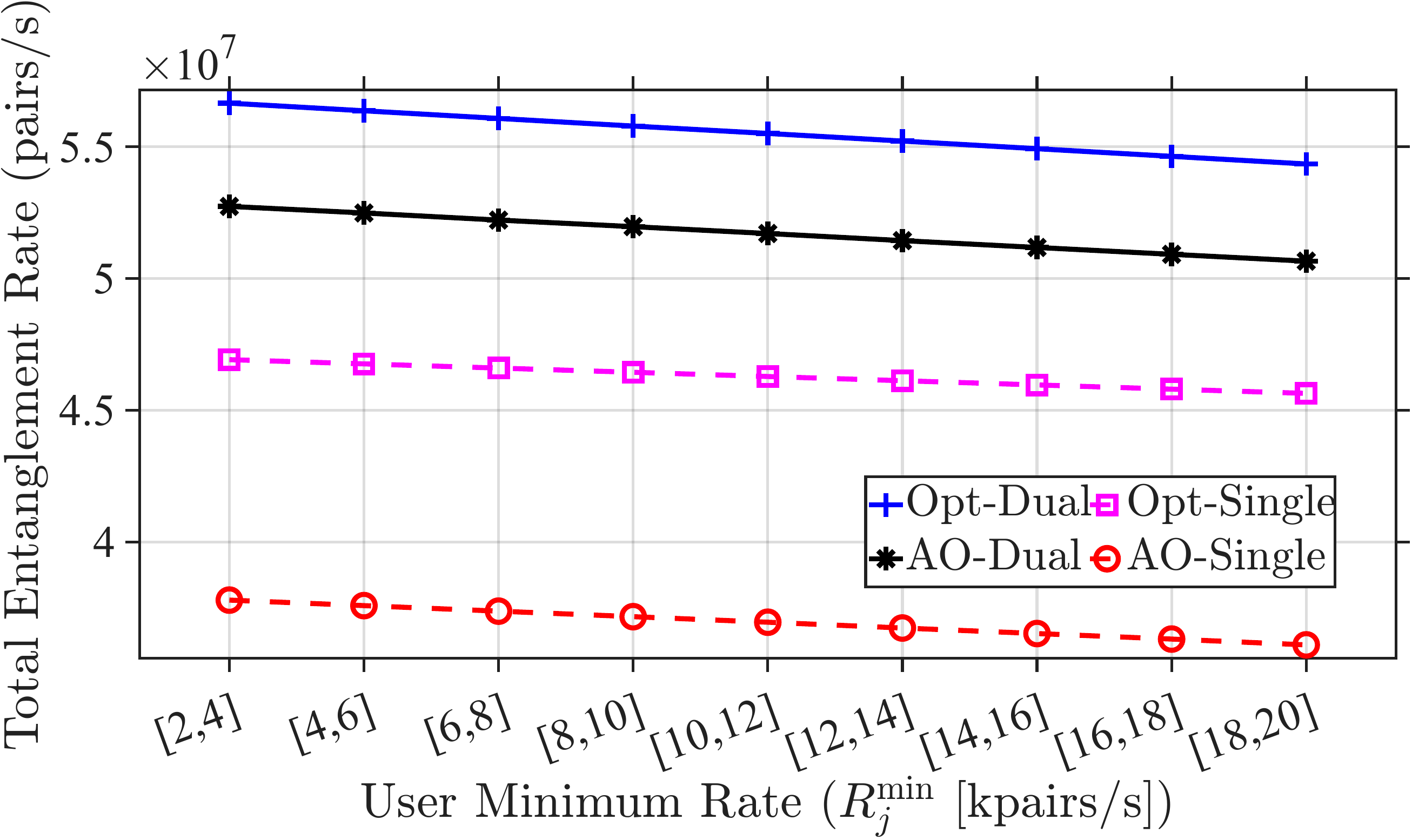}
\caption{Total entanglement rate vs. different $R_j^{\mathrm{min}}$ ranges when $N=10$ and $U=20$.}
    \label{rmin}
\end{figure}

%We evaluate the performance of the proposed DC-enabled QN by comparing it with an SC QN under identical system parameters. We also benchmark the proposed AO algorithm against the optimal solution obtained via Gurobi. As shown in Fig.~\ref{fig:three_plots}, we consider four cases for each experiment: 1) \textit{Optimal} DC, 2) AO-based DC, 3) \textit{Optimal} SC, and 4) AO-based SC.

%Fig.~\ref{fig:users} shows the entanglement rate versus the number of QUs for $N = 10$. The rate increases with more users due to better BS capacity utilization, but gains diminish beyond $U = 40$ as resources are mostly used to meet minimum rate requirements. The DC curves consistently outperform the SC curves, demonstrating the advantage of the proposed DC QN.

Fig.~\ref{users} illustrates how the total entanglement rate varies with the number of QUs when $N = 10$. It can be observed that the entanglement rate increases with the number of QUs across all curves, due to improved utilization of QBS capacity. However, beyond $U = 40$, the performance gain diminishes as QBS resources are increasingly consumed to satisfy the minimum entanglement rate constraints, leaving limited capacity for further entanglement rate maximization. Importantly, the DC curves consistently outperform the SC curves, demonstrating the superiority of the proposed DC-based network.

In Fig.~\ref{bs}, we depict the impact of the number of QBSs on the total entanglement rate when $U = 20$. As can be seen at $N = 2$, all four curves overlap, since the two available QBSs must serve all QUs, effectively reducing the DC architecture to the SC case. As $N$ increases, the curves gradually diverge, resulting in higher achievable entanglement rates. This improvement stems from the increased availability of QBS resources, which alleviates congestion and enables a more flexible allocation of entanglement generation capacity among QUs. Moreover, the DC-based schemes consistently outperform the SC-based schemes because each QU can be simultaneously served by up to two QBSs, allowing for load balancing and more efficient utilization of the  QBS capacities. %As a result, the DC curves continue to increase with $N$. In contrast, the SC curves exhibit saturation when $N \approx U$, since adding more QBSs beyond this point does not further enhance the entanglement rate under the single-connectivity constraint.

Fig.~\ref{rmin} demonstrates the variations of the total entanglement rate with respect to different ranges of minimum rate requirements for fixed values of $N = 10$ and $U = 20$. It can be observed that increasing the minimum entanglement rate requirements leads to a consistent reduction in the achieved entanglement rate, as a larger portion of the available QBS capacity must be reserved to satisfy these constraints. Moreover, DC configurations consistently achieve higher entanglement rates than SC configurations, owing to the ability of each QU to connect to up to two QBSs.

Figs~\ref{users}, \ref{bs}, and~\ref{rmin} show that the proposed AO method achieves competitive, near‑optimal performance, with average optimality gaps ranging from $5\%$ to $19\%$ across different scenarios, indicating an effective trade-off between solution accuracy and computational complexity. In addition, these figures confirm that the proposed DC architecture improves the achievable entanglement rate by $19.5\%$ to $37\%$ on average, compared with the SC approach, highlighting the DC-based strategy as a promising architecture for quantum networks.

\section{Conclusion}
\label{Conclusion}
We  have investigated entanglement rate maximization in DC-enabled wireless quantum networks with multiple QBSs. We formulate a joint QBS–QU association and entanglement generation rate allocation problem as an MINLP, explicitly accounting for QBS generation capacity constraints, QU fidelity and minimum entanglement rate requirements, and impairments of the FSO channel. To efficiently address this challenging problem, we propose an AO framework that decomposes the original MINLP into two tractable linear subproblems. Simulation results demonstrate that the DC architecture consistently outperforms the SC counterpart, while the proposed AO approach achieves near-optimal performance with significantly reduced computational complexity.

\bibliographystyle{IEEEtran}
\bibliography{ref}
\end{document}